\documentclass[11pt]{article}
\usepackage[margin=1.1in]{geometry}
\usepackage{amsmath,amssymb,amsthm}
\usepackage{algorithm}
\usepackage{algpseudocode}
\usepackage{booktabs}
\usepackage{graphicx}
\usepackage{microtype}
\usepackage{rotating}
\usepackage[section]{placeins}
\usepackage[round]{natbib}
\usepackage[colorlinks=true,linkcolor=blue,citecolor=blue,urlcolor=blue]{hyperref}
\usepackage{url}
\usepackage{orcidlink}
\numberwithin{equation}{section}
\newcommand{\E}{\mathbb{E}}
\newcommand{\R}{\mathbb{R}}

\newcommand{\ipmu}[2]{\langle #1,\, #2\rangle_{\mu}}
\newcommand{\ipP}[2]{\langle #1,\, #2\rangle_{P}}
\newcommand{\norm}[1]{\lVert #1\rVert}
\newcommand{\normmu}[1]{\lVert #1\rVert_{\mu}}
\newcommand{\normP}[1]{\lVert #1\rVert_{P}}
\newcommand{\Ltwo}{L^{2}(\mu)}
\newcommand{\LtwoP}{L^{2}(P)}
\newcommand{\dict}{\Phi}
\newcommand{\spanH}{\mathcal{H}}
\newcommand{\proj}[1]{\Pi_{#1}}
\newcommand{\F}{\mathbf{F}}
\newcommand{\Y}{\mathbf{Y}}
\newcommand{\K}{\mathbf{K}}
\newcommand{\bb}{\mathbf{b}}
\newcommand{\bhat}{\widehat{\mathbf{b}}}
\newcommand{\Cmat}{\mathbf{C}}
\newcommand{\Dmat}{\mathbf{D}}
\newcommand{\Gmat}{\mathbf{G}}
\newcommand{\Imat}{\mathbf{I}}
\newcommand{\Hmat}{\mathbf{H}}
\newcommand{\Tmat}{\mathbf{T}}
\newcommand{\pen}{W}

\newcommand{\cum}[1]{\gamma_{#1}}
\newcommand{\diag}{\operatorname{diag}}
\newcommand{\cond}{\operatorname{cond}}

\theoremstyle{plain}
\newtheorem{theorem}{Theorem}[section]
\newtheorem{proposition}[theorem]{Proposition}
\newtheorem{corollary}[theorem]{Corollary}

\theoremstyle{definition}

\newtheorem{assumption}[theorem]{Assumption}
\theoremstyle{remark}
\newtheorem{remark}[theorem]{Remark}

\title{Exact Computation of Non-Gaussian Mismatch Penalties in
Wiener-Hermite Cross-Correlation Identification}
\author{%
  Serhii Zabolotnii\,\orcidlink{0000-0003-0242-2234}\\[3pt]
  \small Department of Information, Multimedia Technologies and Design,\\
  \small Cherkasy State Business College, Cherkasy 18028, Ukraine\\[2pt]
  \small State Scientific Research Institute of Armament and Military Equipment\\
  \small Testing and Certification, Cherkasy, Ukraine\\[2pt]
  \small Department of Cybernetics and Applied Mathematics,\\
  \small Uzhhorod National University, Uzhhorod, Ukraine\\[3pt]
  \small\texttt{zabolotnii.serhii@csbc.edu.ua}%
}
\date{}

\begin{document}
\maketitle

\begin{abstract}
\noindent
Wiener-Hermite cross-correlation identification represents a polynomial
response in the Hermite basis. Under Gaussian excitation the basis is
orthogonal and a diagonal rule recovers it exactly; under non-Gaussian
excitation the same basis is kept, but its Gram matrix gains off-diagonal
terms and the diagonal rule is no longer the population projection. We give
the exact finite-order excess $L^{2}(P)$ risk of this mismatch: a moment
quadratic form from two Hankel-Cholesky factorizations and one diagonal
solve, at $O(s^{3})$ cost from moments to order $2s$. Closed cumulant
forms at orders three and four expose which non-Gaussian features drive it;
symmetry protects the Gaussian basis only through order two. A
bootstrap decides, from data, whether a matched basis is worth building; on a
Wiener-Hammerstein benchmark it separates a near-Gaussian channel
(penalty $\approx10^{-4}$) from a skewed output (penalty $0.05$). The
computation is a weighted-$L^{2}$ projection whose core normal-system
correspondence is machine-checked in Lean~4.
\end{abstract}



\section{Introduction}\label{sec:intro}

Polynomial input--output models are often identified by cross-correlation. In the
classical Wiener setting one observes an excitation $X$ and a response
$Y=f(X)+\varepsilon$, with $f$ approximated in
$\Pi_s=\operatorname{span}\{1,x,\dots,x^s\}$. If
$X$ is Gaussian with variance $\sigma^2$, the normalized Hermite polynomials
$\eta_0,\dots,\eta_s$ are orthonormal in the corresponding $L^2$ space, and the
Lee--Schetzen diagonal rule estimates each coefficient by correlating $Y$ with
one basis function. At the population level this is simply the orthogonal
projection of $f$ onto $\Pi_s$; no off-diagonal linear solve is needed because
the Gram matrix of the Wiener--Hermite family is diagonal
\citep{lee1965measurement,schetzen1980volterra}.

The computational difficulty begins when the excitation is not Gaussian. Let
$P$ be the true mean-zero input law, still with variance $\sigma^2$, and let
$Q=N(0,\sigma^2)$ be the Gaussian law used to tabulate the Hermite basis. The
same diagonal cross-correlation rule remains attractive because it is cheap,
stable, and already implemented in many Volterra--Wiener identification workflows. But
the Hermite basis is no longer orthogonal in $L^2(P)$. Its Gram matrix acquires
off-diagonal entries, so the diagonal rule is no longer the orthogonal projection
under the law that generated the data. The resulting error is not a truncation
residual--we keep $f\in\Pi_s$ throughout--but the deterministic price of using a
Gaussian basis under a non-Gaussian input.

We compute that price exactly. Let
$\psi_0,\dots,\psi_s$ be the orthonormal polynomial basis matched to $P$, and
write the signal as $f=\sum_{k\le s}b_k\psi_k$. Let
$\eta_0,\dots,\eta_s$ be the variance-matched Gaussian Hermite basis. If
$\Cmat_{ki}=\langle\eta_k,\psi_i\rangle_P$ and
$\Dmat=\operatorname{diag}(\Cmat\Cmat^\top)$, then the excess $L^2(P)$ risk of
the diagonal Wiener--Hermite reconstruction is
\begin{equation}\label{eq:penalty-intro}
  \pen(s;\bb,P)=
  \bigl\|(\Imat-\Cmat^\top\Dmat^{-1}\Cmat)\bb\bigr\|_2^2 .
\end{equation}
Theorem~\ref{thm:penalty} proves \eqref{eq:penalty-intro}. Computationally, the
formula requires only the raw moments of $P$ up to order $2s$: build the Hankel
moment matrix of $P$, build the variance-matched Gaussian Hankel matrix, take
two Cholesky factors, and perform one diagonal solve. The cost is $O(s^3)$ and
the output is an exact population penalty once the moments are specified.

The formula answers a practical basis-selection question. For a fixed order
$s$, should one keep the tabulated Gaussian basis or estimate a
distribution-matched basis from data? Small off-diagonal Hermite correlations
under $P$ certify that the Gaussian basis is adequate; large off-diagonal mass
means that the diagonal rule has a non-removable error floor. With empirical
moments, the same computation yields a plug-in diagnostic. Bootstrap intervals
then quantify whether the estimated penalty is resolvable at the available
sample size.

\subsection{Classical background and the non-Gaussian failure mode}
\label{sec:wiener-background}

The Wiener--Hermite construction is natural only because Gaussian inputs make
the Hermite polynomials orthogonal. In the Gaussian case, for basis functions
$\eta_i$ and $\eta_j$,
\[
  \langle\eta_i,\eta_j\rangle_Q=0\quad (i\ne j),
\]
and the coefficient of $\eta_k$ is recovered by the diagonal ratio
\[
  \widehat b_k=\frac{\langle f,\eta_k\rangle_Q}
                    {\langle\eta_k,\eta_k\rangle_Q}.
\]
Replacing $Q$ by a non-Gaussian $P$ changes only the inner product, but that is
enough to change the computation. The same polynomials now satisfy
\[
  \langle\eta_i,\eta_j\rangle_P\ne0
  \quad\text{for some }i\ne j,
\]
so coefficient recovery requires the full Gram system. The diagonal
cross-correlation rule ignores those off-diagonal terms. The mismatch penalty
\eqref{eq:penalty-intro} is precisely the squared residual of that diagonal
shortcut, measured under $P$.

This distinction also separates the present problem from standard approximation
questions. Generalized polynomial chaos, distribution-matched Volterra
orthogonalization, and sieve-type series estimation study the truncation or
projection error of a matched expansion
\citep{yasui1979stochastic,korenberg1996identification,ghanem1991stochastic,xiu2002wiener,ernst2012convergence,chen2007large}.
Here the signal already belongs to the finite polynomial span. The only question
is the computational error caused by keeping the Gaussian diagonal rule when the
input law is not Gaussian.

\subsection{Contributions}

The main contribution is the finite-order moment algorithm for
\eqref{eq:penalty-intro}. More explicitly:
\begin{enumerate}
\item We derive the exact order-$s$ mismatch penalty for the
  Wiener--Hermite diagonal cross-correlation rule under a non-Gaussian input.
  The penalty is a computable moment quadratic form, not an asymptotic or
  simulation estimate.
\item We give closed cumulant forms at orders three and four. These formulas
  show which cumulant defects activate the penalty. In particular, symmetry
  protects the Gaussian basis at order two but not at order three when excess
  kurtosis is present.
\item We use the formula as a basis-selection diagnostic. The diagnostic uses
  empirical moments, an off-diagonal Gram certificate, and bootstrap uncertainty
  to decide whether the Gaussian basis should be replaced by a matched basis.
\item We verify the formula by independent computational paths: symbolic moment
  algebra, adaptive quadrature, randomized quasi--Monte Carlo, and
  finite-sample identification. On a real Wiener--Hammerstein benchmark, the
  diagnostic gives near-zero penalties for the designed excitation and nonzero
  penalties for the skewed circuit output.
\item As a secondary organizing layer, a generating-element projection calculus
  relates the same normal-system computation to estimation, detection,
  classification, and characteristic-function distance. The Lean~4 module
  machine-checks the core projection correspondence as certification, not as a
  separate mathematical contribution.
\end{enumerate}

\subsection{Relation to prior work}
\label{sec:related}

Orthogonal expansions for nonlinear systems originate with Wiener's homogeneous
chaos \citep{wiener1938homogeneous} and the Fourier--Hermite functionals of
\citet{cameron1947orthogonal}. The Lee--Schetzen rule
\citep{lee1965measurement,schetzen1980volterra} is the classical
cross-correlation identification method built on that Gaussian orthogonality.
When the input is non-Gaussian, the principled matched construction replaces the
Hermite family by orthonormal polynomials of the actual input law
\citep{yasui1979stochastic,korenberg1996identification}. The generalized
polynomial-chaos literature makes the same move through the Wiener--Askey scheme
\citep{ghanem1991stochastic,xiu2002wiener}. Those theories justify matched
expansions; they do not give a finite-order expression for the cost of retaining
the Gaussian diagonal rule instead of matching the basis.

The order-two version of this mismatch problem appears in the companion
Volterra--Wiener--Kunchenko manuscript \citep{zabolotnii2026vwk}. The present
paper treats that result as the base case and supplies the general finite-order
formula, the order-three and order-four cumulant forms, and the computational
diagnostic. The cumulant identities also connect to the polynomial maximization
method of Kunchenko \citep{kunchenko2002polynomial,kunchenko2006stochastic},
where the same cumulant ratios appear as variance-reduction coefficients rather
than mismatch penalties.

The projection layer in Section~\ref{sec:calculus} belongs to the broader
estimating-function and GMM tradition. Quasi-score theory treats estimation as
projection in an inner-product space
\citep{godambe1987quasi,heyde1997quasi,small1994hilbert}; GMM information is a
projection norm \citep{hansen1982large}; redundancy and grid-refinement
monotonicity are classical in moment-condition and empirical-characteristic
function methods \citep{breusch1999redundancy,feuerverger1981efficiency}. We use
that geometry rather than claiming it. The new part is the exact mismatch penalty
and its computational consequences for Wiener--Hermite cross-correlation
identification under non-Gaussian excitation.

\subsection{Organization}

Section~\ref{sec:penalty} derives the matched and mismatched bases, the exact
matrix formula, the $O(s^3)$ algorithm, and the closed cumulant forms.
Section~\ref{sec:calculus} records the projection calculus that organizes the
normal systems used by the related method families. Section~\ref{sec:identities}
derives the cross-branch cumulant identities and the plug-in diagnostic.
Section~\ref{sec:verification} verifies the formulas numerically and reports the
Wiener--Hammerstein diagnostic. Section~\ref{sec:discussion} states limitations
and extensions.

\section{Exact order-$s$ mismatch penalties}\label{sec:penalty}

\subsection{Matched and mismatched bases}\label{sec:setup-penalty}

Let $P$ be a law on $\R$ with mean zero, variance $\sigma^{2}>0$, and finite
moments to order $2s$, and assume its support contains at least $s+1$ points;
equivalently, the Hankel moment matrix $\Hmat_s=(m_{i+j})_{i,j=0}^{s}$ is positive
definite (this non-degeneracy is the standing hypothesis of the section). Let
$\Pi_s=\operatorname{span}\{1,x,\dots,x^{s}\}\subset\LtwoP$. Under it, oriented
Gram--Schmidt in $\LtwoP$ (positive leading coefficients) produces the
\emph{matched} orthonormal family $\psi_0,\dots,\psi_s$; equivalently
$\Tmat_s\Hmat_s\Tmat_s^{\top}=\Imat$ with $\Hmat_s$ the (now invertible) Hankel
moment matrix and $\Tmat_s$ the inverse Cholesky factor
\citep{zabolotnii2026vwk}. The \emph{Wiener} family
$\eta_0,\dots,\eta_s$ is the same construction under the variance-matched
Gaussian law $Q=N(0,\sigma^{2})$, i.e.\ normalized Hermite polynomials. Classical
cross-correlation identification \citep{lee1965measurement,schetzen1980volterra}
keeps the Wiener family and estimates each coefficient by the \emph{diagonal
rule}
\begin{equation}\label{eq:diagonal-rule}
  \bhat_k\;=\;\frac{\ipP{f}{\eta_k}}{\ipP{\eta_k}{\eta_k}},
  \qquad k=0,\dots,s,
\end{equation}
which ignores the off-diagonal Gram structure that $\{\eta_k\}$ acquires under
$P\ne Q$. Each $\eta_k$ is a polynomial of exact degree $k$, so the standing
non-degeneracy gives $\ipP{\eta_k}{\eta_k}>0$ and makes the diagonal rule
\eqref{eq:diagonal-rule} well-defined; and since $\{\psi_i\}$ and $\{\eta_k\}$ are
two bases of $\Pi_s$, the cross-coordinate matrix $\Cmat$ introduced below is
square and nonsingular. Write the unknown response in matched coordinates,
$f=\sum_{k\le s}b_k\psi_k$. Since $f\in\Pi_s$, the matched projection recovers
$f$ exactly; the entire excess $\LtwoP$ risk of the Wiener reconstruction
$\widehat f=\sum_k\bhat_k \eta_k$,
\[
  \pen(s;\bb,P)\;=\;\bigl\lVert f-\widehat f\bigr\rVert_{P}^{2},
\]
is therefore the price of keeping the Gaussian basis. Because the signal
belongs to the span, this excess risk is created entirely by the mis-projection
\eqref{eq:diagonal-rule}; it is not a truncation residual.

\subsection{The matrix form}

\begin{theorem}[exact order-$s$ penalty]\label{thm:penalty}
Under the setting above---in particular the standing non-degeneracy
($\Hmat_s\succ0$), which makes $\Cmat$ square and nonsingular and $\Dmat$ a
strictly positive diagonal, so $\Dmat^{-1}$ exists---with
$\Cmat_{ki}=\ipP{\eta_k}{\psi_i}$ and $\Dmat=\diag(\Cmat\Cmat^{\top})$,
\begin{equation}\label{eq:penalty-matrix}
  \pen(s;\bb,P)\;=\;
  \bigl\lVert(\Imat-\Cmat^{\top}\Dmat^{-1}\Cmat)\,\bb\bigr\rVert_{2}^{2}.
\end{equation}
Every entry of $\Cmat$ is a polynomial moment of $P$ of order at most $2s$.
The penalty $\pen(s;\cdot\,,P)$ vanishes identically if and only if the Gram matrix
$\Gmat=\Cmat\Cmat^{\top}$ of the Wiener family under $P$ is diagonal, i.e.\ the
Gaussian-basis functions remain $P$-orthogonal.
\end{theorem}

\begin{proof}
Both families span $\Pi_s$, so $\eta_k=\sum_i\Cmat_{ki}\psi_i$ with
$\Cmat_{ki}=\ipP{\eta_k}{\psi_i}$, and
$\Gmat_{kj}=\ipP{\eta_k}{\eta_j}=(\Cmat\Cmat^{\top})_{kj}$, whence
$\Dmat_{kk}=\ipP{\eta_k}{\eta_k}$. The numerator of \eqref{eq:diagonal-rule} is
$\ipP{f}{\eta_k}=(\Cmat\bb)_k$, so $\bhat=\Dmat^{-1}\Cmat\bb$ and, in matched
coordinates, $\widehat f=\Cmat^{\top}\bhat=\Cmat^{\top}\Dmat^{-1}\Cmat\,\bb$.
Orthonormality of $\{\psi_i\}$ turns the $\LtwoP$ norm into the Euclidean norm
of coordinates, giving \eqref{eq:penalty-matrix}. If $\Gmat$ is diagonal, the
diagonal rule is the orthogonal projection onto $\Pi_s$ and the penalty
vanishes for every $\bb$; conversely $\pen\equiv0$ forces
$\Cmat^{\top}\Dmat^{-1}\Cmat=\Imat$, which after multiplying by $\Cmat$ on the
left yields $\Gmat\Dmat^{-1}\Cmat=\Cmat$, hence $\Gmat=\Dmat$.
\end{proof}

\begin{remark}[practical recipe and cost]\label{rem:recipe}
Evaluating \eqref{eq:penalty-matrix} requires: moments of $P$ to order $2s$; two
Cholesky factorizations of $(s{+}1)\times(s{+}1)$ Hankel matrices (matched and
Gaussian); and one diagonal solve---$O(s^{3})$ floating-point operations in
total. No simulation, quadrature, or asymptotics enters; given the moments, the
penalty is exact. By contrast, Edgeworth-type corrections are asymptotic in
sample size and bootstrap or Monte Carlo assessments are stochastic;
Section~\ref{sec:verification} uses them only as independent verification
paths.
\end{remark}

\begin{algorithm}[t]
\caption{Mismatch penalty $\pen(s;\bb,P)$ from raw moments}\label{alg:penalty}
\begin{algorithmic}[1]
\Require moments $m_0,\dots,m_{2s}$ of $P$ ($m_0{=}1$, $m_1{=}0$, $m_2{=}\sigma^{2}$); signal $\bb\in\R^{s+1}$ in matched coordinates
\Ensure penalty $\pen(s;\bb,P)$
\State $\Hmat\gets(m_{i+j})_{i,j=0}^{s}$ \Comment{Hankel moment matrix of $P$}
\State $\Tmat\gets\operatorname{chol}(\Hmat)^{-1}$ \Comment{rows are the matched basis $\psi_0,\dots,\psi_s$; $\Tmat\Hmat\Tmat^{\top}=\Imat$}
\State $\mathbf{U}\gets\operatorname{chol}(\Hmat^{Q})^{-1}$, with $\Hmat^{Q}$ the Hankel matrix of $N(0,\sigma^{2})$ \Comment{rows are the Wiener basis $\eta_0,\dots,\eta_s$}
\State $\Cmat\gets\mathbf{U}\,\Hmat\,\Tmat^{\top}$ \Comment{$\Cmat_{ki}=\ipP{\eta_k}{\psi_i}$; only moments to order $2s$ enter}
\State $\Dmat\gets\diag(\Cmat\Cmat^{\top})$ \Comment{$\Dmat_{kk}=\ipP{\eta_k}{\eta_k}$}
\State $\widehat\bb\gets\Dmat^{-1}\Cmat\,\bb$ \Comment{diagonal cross-correlation rule \eqref{eq:diagonal-rule}}
\State \Return $\lVert\bb-\Cmat^{\top}\widehat\bb\rVert_{2}^{2}$
\end{algorithmic}
\end{algorithm}

Algorithm~\ref{alg:penalty} makes Remark~\ref{rem:recipe} explicit and fixes the
normalization conventions: both bases use the oriented inverse-Cholesky factor
(rows of $\Tmat$ for the matched $\psi_k$, rows of $\mathbf{U}$ for the Wiener
$\eta_k$), so the cross-Gram $\Cmat=\mathbf{U}\Hmat\Tmat^{\top}$ and its diagonal
$\Dmat$ are obtained in $O(s^{3})$ operations and the penalty is the squared
residual of the diagonal rule. The off-diagonal mass of $\Gmat=\Cmat\Cmat^{\top}$
supplies an even cheaper certificate.

\begin{proposition}[off-diagonal diagnostic bound]\label{prop:bound}
Let $\Gmat=\Cmat\Cmat^{\top}$ be the Gram matrix of the Wiener family under $P$,
$\Dmat=\diag(\Gmat)$, and $\widehat\Gmat=\Dmat^{-1/2}\Gmat\Dmat^{-1/2}$ its
correlation-normalized form (unit diagonal). Write $\mathbf{E}=\widehat\Gmat-\Imat$
and $\varepsilon=\lVert\mathbf{E}\rVert_{2}$ for the spectral off-diagonal mass.
With $\mathbf{M}=\Dmat^{-1/2}\Cmat$ (so $\mathbf{M}\mathbf{M}^{\top}=\widehat\Gmat$),
\begin{equation}\label{eq:bound-exact}
  \pen(s;\bb,P)=\bb^{\top}\mathbf{M}^{\top}\mathbf{E}\,\widehat\Gmat^{-1}\mathbf{E}\,\mathbf{M}\,\bb,
\end{equation}
and whenever $\varepsilon<1$,
\begin{equation}\label{eq:bound-ineq}
  \pen(s;\bb,P)\;\le\;\frac{1+\varepsilon}{1-\varepsilon}\,\varepsilon^{2}\,\lVert\bb\rVert_{2}^{2}.
\end{equation}
The off-diagonal mass obeys $\varepsilon\le\max_{k}\sum_{j\ne k}\lvert\widehat\Gmat_{kj}\rvert$, a
Gershgorin row sum of the Gram correlations. A single pass over those
correlations---no residual operator and no signal $\bb$---therefore certifies
adequacy of the Gaussian basis: when they are small the penalty per unit signal
energy is at most $\approx\varepsilon^{2}$. The bound is informative in this
near-orthogonal regime and becomes vacuous once the Wiener family is strongly
non-$P$-orthogonal ($\varepsilon\ge1$), where the exact form
\eqref{eq:penalty-matrix} is needed.
\end{proposition}

\begin{proof}
Because $\Cmat$ is square and invertible, $\Cmat^{\top}\Gmat^{-1}\Cmat=\Imat$, so
the residual operator of \eqref{eq:penalty-matrix} is
$\Imat-\Cmat^{\top}\Dmat^{-1}\Cmat=\Cmat^{\top}(\Gmat^{-1}-\Dmat^{-1})\Cmat
=-\Cmat^{\top}\Gmat^{-1}(\Gmat-\Dmat)\Dmat^{-1}\Cmat$. Conjugating by
$\Dmat^{1/2}$ and using $\mathbf{M}=\Dmat^{-1/2}\Cmat$,
$\mathbf{M}\mathbf{M}^{\top}=\widehat\Gmat$, and
$\Dmat^{-1/2}(\Gmat-\Dmat)\Dmat^{-1/2}=\mathbf{E}$ rewrites it as
$-\mathbf{M}^{\top}\widehat\Gmat^{-1}\mathbf{E}\,\mathbf{M}$; squaring and simplifying with
$\widehat\Gmat^{-1}\mathbf{M}\mathbf{M}^{\top}\widehat\Gmat^{-1}=\widehat\Gmat^{-1}$ gives
\eqref{eq:bound-exact}. For \eqref{eq:bound-ineq},
$\lVert\mathbf{E}\widehat\Gmat^{-1}\mathbf{E}\rVert_{2}\le\varepsilon^{2}/\lambda_{\min}(\widehat\Gmat)$
and $\lambda_{\max}(\mathbf{M}^{\top}\mathbf{M})=\lambda_{\max}(\widehat\Gmat)$, while the
unit diagonal gives $\lambda_{\max}(\widehat\Gmat)\le1+\varepsilon$ and
$\lambda_{\min}(\widehat\Gmat)\ge1-\varepsilon$. The Gershgorin estimate is the
standard spectral-norm bound for a symmetric matrix with zero diagonal.
\end{proof}

\subsection{Closed cumulant forms}\label{sec:closed-forms}

Standardize $\sigma^2=1$ and write $\cum{k}=\kappa_k(P)/\sigma^{k}$ for the
standardized cumulants of $P$ and $h_k$ for the monic Hermite polynomials.
(The Wiener basis functions are written $\eta_k$ to keep them distinct from the
PMM variance-reduction coefficients $g_S$.) The second-order
case is the anchor. Here $\pen(2)=(\gamma_{2}\lambda)^{2}+(\gamma_{2}\rho-b_2)^{2}$
with $\lambda=\cum{3}$, $\rho^{2}=\cum{4}+2-\cum{3}^{2}$,
$\gamma_{2}=(b_1\lambda+b_2\rho)/(\lambda^{2}+\rho^{2})$, vanishing iff
$\cum{3}=0$; this is Proposition~3 of \citet{zabolotnii2026vwk} and is not
re-proved here.

\begin{theorem}[closed forms at $s=3,4$]\label{thm:closed-forms}
Let $P$ be symmetric ($\cum{3}=\cum{5}=0$) and, for the order $s\in\{3,4\}$ in
question, have finite moments through order $2s$ with Hankel matrix $\Hmat_s$
positive definite (equivalently, support on at least $s+1$ points); in particular
the order-4 even block below requires a finite eighth moment and $\Hmat_4\succ0$.
Then:
\begin{enumerate}
\item[(i)] (order 3)
$\displaystyle
  \pen(3;\bb,P)=\frac{\cum{4}^{2}\,(b_1^{2}+b_3^{2})}{6+9\cum{4}+\cum{6}}
  \;=\;(1-g_3)\,(b_1^{2}+b_3^{2}),
$
where $g_3$ is the third-order PMM variance-reduction coefficient;
\item[(ii)] (order 4, parity decoupling) $\pen(4;\bb,P)$ splits into odd and
even blocks with no cross terms; the odd block equals the order-3 expression in
(i) verbatim, and the even block is structured entirely by Hermite cross-moments
under $P$:
\[
  \pen_{\mathrm{even}}(4)
  \;=\;\frac{\bigl(b_0\,c_0\sqrt{d_2}+b_2\,c_2\bigr)^{2}
        +b_4^{2}\,\bigl(c_0^{2}d_2+c_2^{2}\bigr)}
       {d_2\,\normP{h_4}^{2}},
\]
where $c_0=\E_P h_4=\cum{4}$, $c_2=\ipP{h_2}{h_4}=8\cum{4}+\cum{6}$,
$d_2=\normP{h_2}^{2}=\cum{4}+2$, and
$\normP{h_4}^{2}=24+72\cum{4}+35\cum{4}^{2}+16\cum{6}+\cum{8}$. (Equality with
the raw symbolic expansion is itself a symbolic identity; see
Appendix~\ref{app:even-block}.)
\end{enumerate}
For skewed $P$ at $s=3$ a closed form exists in monic coordinates as a rational
function of $(\cum{3},\dots,\cum{6})$ but does not factor compactly
($\approx2300$ operations); Appendix~\ref{app:general-s3} records it, and the
matrix form \eqref{eq:penalty-matrix} remains the recommended computation.
\end{theorem}

\begin{proof}
For symmetric $P$ all odd-order moments vanish, so $\Cmat$ block-diagonalizes
into an odd part on $\{x,x^{3}\}$ and an even part on $\{1,x^{2},x^{4}\}$; the two
blocks contribute additively to \eqref{eq:penalty-matrix}, and we treat them in
turn.

\emph{Odd block (orders $3$ and $4$).} The matched odd functions are $\psi_1=x$
and $\psi_3=(x^{3}-\mu_4 x)/\sqrt{\mu_6-\mu_4^{2}}$; the Wiener odd functions are
$\eta_1=x$ and $\eta_3=(x^{3}-3x)/\sqrt6$. Orthonormality of the matched family
fixes the cross-Gram entries $\Cmat_{11}=\ipP{\eta_1}{\psi_1}=1$,
$\Cmat_{13}=\ipP{\eta_1}{\psi_3}=0$, and
$q:=\Cmat_{31}=\ipP{\eta_3}{\psi_1}=(\mu_4-3)/\sqrt6=\cum{4}/\sqrt6$,
$r:=\Cmat_{33}=\ipP{\eta_3}{\psi_3}$, with
$d:=\ipP{\eta_3}{\eta_3}=(\mu_6-6\mu_4+9)/6=q^{2}+r^{2}$, so $\Dmat=\diag(1,d)$.
The residual operator of \eqref{eq:penalty-matrix} restricted to the odd block is
then
\[
  \Imat-\Cmat^{\top}\Dmat^{-1}\Cmat
  =\frac1d\begin{pmatrix}-q^{2}&-qr\\[2pt]-qr&q^{2}\end{pmatrix},
\]
whose squared action on $(b_1,b_3)$ is
$\tfrac{q^{2}}{d^{2}}\bigl[(qb_1+rb_3)^{2}+(rb_1-qb_3)^{2}\bigr]
=\tfrac{q^{2}}{d}\,(b_1^{2}+b_3^{2})$. Substituting $q^{2}=\cum{4}^{2}/6$ and
$d=(\mu_6-6\mu_4+9)/6=(6+9\cum{4}+\cum{6})/6$ gives
$\pen(3)=\cum{4}^{2}(b_1^{2}+b_3^{2})/(6+9\cum{4}+\cum{6})$, which equals
$(1-g_3)(b_1^{2}+b_3^{2})$ by the definition of the third-order PMM coefficient
\citep{kunchenko2002polynomial}. This proves~(i). The odd dictionary
$\{x,x^{3}\}$ and the odd Wiener pair are unchanged from $s=3$ to $s=4$, so this
block is identical at order $4$---the odd-block invariance asserted in~(ii).

\emph{Even block (order $4$).} The even Wiener triple is $\{1,h_2,h_4\}$ with
$h_2=x^{2}-1$ and $h_4=x^{4}-6x^{2}+3$. Two simplifications follow from symmetry.
First, $\E_P h_2=m_2-1=0$ and $\ipP{h_2}{x}=0$, so $h_2$ is $P$-proportional to
the matched quadratic $\psi_2$; together with $\eta_0=\psi_0=1$ this means the
diagonal rule recovers the constant and quadratic coefficients \emph{exactly}.
The entire even-block penalty is therefore the error of the $h_4$ reconstruction
alone. Second, $h_4$ fails $P$-orthogonality to the constant and the quadratic
through exactly two defects, the kurtosis defect $c_0=\E_P h_4=\cum{4}$ and the
cross-moment $c_2=\ipP{h_2}{h_4}=8\cum{4}+\cum{6}$, so the diagonal estimate of
the $h_4$ coefficient is contaminated by $b_0$ and $b_2$ while $b_4$ is itself
mis-normalized. Writing $d_2=\normP{h_2}^{2}=\cum{4}+2$ and
$\normP{h_4}^{2}=24+72\cum{4}+35\cum{4}^{2}+16\cum{6}+\cum{8}$ for the two
normalizers, the squared $\LtwoP$ norm of that residual collects into
\[
  \pen_{\mathrm{even}}(4)
  =\frac{\bigl(b_0\,c_0\sqrt{d_2}+b_2\,c_2\bigr)^{2}+b_4^{2}\bigl(c_0^{2}d_2+c_2^{2}\bigr)}
        {d_2\,\normP{h_4}^{2}} ,
\]
the cross term $(b_0c_0\sqrt{d_2}+b_2c_2)$ being the contamination of the $h_4$
coefficient and $b_4^{2}(c_0^{2}d_2+c_2^{2})$ its self-normalization defect.
Adding the two blocks proves~(ii). That this structured form equals the raw
residual is a polynomial identity in the cumulants; it is carried out in rational
coordinates and confirmed to vanish exactly in Appendix~\ref{app:even-block},
and both blocks are checked against \eqref{eq:penalty-matrix} to near machine
precision in Section~\ref{sec:verification}.
\end{proof}

\begin{remark}[the skewed case at $s=3$: where skewness enters]\label{rem:skew-s3}
Although the general skewed penalty does not factor, its structure is
explicit. Because $\eta_0=\psi_0=1$ and $\eta_1=\psi_1=x$ for every mean-zero,
unit-variance $P$, the diagonal rule recovers the constant and linear
coordinates \emph{exactly} ($\widehat b_0=b_0$, $\widehat b_1=b_1$), so the whole
$s=3$ penalty is the residual of the quadratic--cubic reconstruction. Writing
$\widehat b_2,\widehat b_3$ for the diagonal estimates of the quadratic and cubic
Wiener coefficients, the residual in matched coordinates is
\[
  r_0=-\frac{\cum{3}}{\sqrt6}\,\widehat b_3,\qquad
  r_1=-\frac{\cum{3}}{\sqrt2}\,\widehat b_2-\frac{\cum{4}}{\sqrt6}\,\widehat b_3,
\]
with $r_2,r_3$ the quadratic and cubic residuals, and
$\pen(3)=r_0^{2}+r_1^{2}+r_2^{2}+r_3^{2}$. Skewness enters \emph{linearly} and
only through $r_0,r_1$: $\cum{3}$ leaks the cubic coefficient into the constant
($r_0$) and the quadratic coefficient into the linear ($r_1$), while $\cum{4}$
leaks the cubic into the linear. At $\cum{3}=0$ one has $r_0=0$ and the skewness
term of $r_1$ vanishes, recovering $\pen(3)=(1-g_3)(b_1^{2}+b_3^{2})$ of
Theorem~\ref{thm:closed-forms}(i). Specializing to $b_2=0$ or $b_1=0$ substitutes
into this exact decomposition; in each case skewness contributes through $r_0,r_1$
alone---the partial interpretability that the unfactored rational form of
Appendix~\ref{app:general-s3} obscures.
\end{remark}

\begin{corollary}[order threshold]\label{cor:threshold}
Let $P$ be symmetric with $\cum{4}\neq0$, not supported on fewer than four
points. Then $\pen(2;\bb,P)=0$ for every $\bb$, while $\pen(3;\bb,P)>0$
whenever $(b_1,b_3)\neq0$. Symmetry of the input protects the Gaussian basis at
second order and no further.
\end{corollary}

\begin{proof}
At $s=2$ the closed form above vanishes because $\lambda=\cum{3}=0$ forces both
summands to zero (equivalently, $\Gmat$ is diagonal at order $2$ for symmetric
$P$, so Theorem~\ref{thm:penalty} applies). At $s=3$,
Theorem~\ref{thm:closed-forms}(i) gives
$\pen(3)=\cum{4}^{2}(b_1^{2}+b_3^{2})/(6+9\cum{4}+\cum{6})$; the numerator is
positive by assumption, and the denominator equals
$\E_P(x^{3}-3x)^{2}=\normP{h_3}^{2}$, which is strictly positive because a
nonzero cubic cannot vanish $P$-almost surely when the support has at least
four points.
\end{proof}

\begin{remark}[finite-support laws]\label{rem:support}
The support hypothesis of Corollary~\ref{cor:threshold} cannot be dropped. The
order-$s$ construction presumes $\operatorname{supp}(P)$ has at least $s+1$
points, so that the Hankel matrix $\Hmat_s$ is nonsingular and $\psi_s$ exists. A
symmetric law on three points $\{-a,0,a\}$ violates this at $s=3$: there
$x^{3}=a^{2}x$ on the support, so the matched cubic degenerates, $\Hmat_3$ is
singular, and $\pen(3)$ is undefined rather than positive. The second-order
protection then appears to extend only because order $3$ is unavailable on three
points, not because the Gaussian basis is spared. With four or more support
points $\Hmat_3$ is nonsingular and the corollary holds verbatim; a nonzero
kurtosis defect $\cum{4}\neq0$ forces $\pen(3)>0$ irrespective of the sign of the
excess kurtosis.
\end{remark}

\begin{remark}[mechanism]\label{rem:mechanism}
The penalty cascades through Hermite cross-moments that are cumulant defects.
For symmetric $P$ all odd cross-moments vanish, so the first nonzero defect is
$\ipP{h_1}{h_3}=\cum{4}$, which activates the odd block at $s=3$; at $s=4$ the
even block joins through $\E_P h_4=\cum{4}$ and $\ipP{h_2}{h_4}$. Skewness
($\cum{3}\ne0$) mixes parities already at $s=2$. Each order recruits exactly the
next cumulant the Gaussian basis fails to annihilate, which explains both
Corollary~\ref{cor:threshold} and the denominators
$\normP{h_k}^{2}$ appearing in Theorem~\ref{thm:closed-forms}.
\end{remark}

Figure~\ref{fig:penalty-growth} shows $\pen(s)$ for $s=2,\dots,5$ across six
laws: zero for the Gaussian control at all orders, the threshold jump at $s=3$
for the symmetric kurtotic laws, and growth driven by skewness from $s=2$ for
the skewed families.

\begin{figure}[t]
\centering
\includegraphics[width=.78\linewidth]{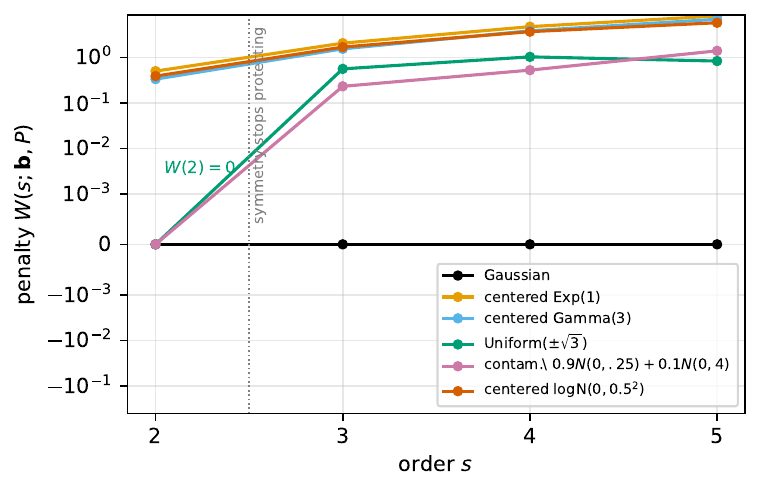}
\caption{Penalty $\pen(s;\bb,P)$ for the canonical signal across six input
laws (symlog scale). Symmetric kurtotic laws (uniform, contaminated normal)
have $\pen(2)=0$ exactly and join at $s=3$ (Corollary~\ref{cor:threshold});
skewed laws pay from $s=2$; the Gaussian control stays at zero.}
\label{fig:penalty-growth}
\end{figure}

\section{The generating-element projection calculus}\label{sec:calculus}

\subsection{Setting and standing assumptions}\label{sec:setting}

Throughout, $(\mathcal{U},\mathcal{A},\mu)$ is a measure space and
$\Ltwo=\Ltwo(\mu;\mathbb{K})$, $\mathbb{K}\in\{\R,\mathbb{C}\}$, carries the
inner product $\ipmu{f}{g}=\int \bar f\,g\,d\mu$, conjugate-linear in its first
argument (so the Gram matrix $\F$ below is Hermitian, $\F=\F^{*}$, and the energy
$\K^{*}\Y$ is real; in the real-valued branches the conjugation is vacuous and
$\F$ is symmetric). We make three standing assumptions, stated here once.

\begin{assumption}\label{ass:dict}
The \emph{dictionary} $\dict_S=(\varphi_0,\dots,\varphi_S)$ is a finite family in
$\Ltwo$. It is \emph{not} assumed orthogonal or linearly independent.
\end{assumption}

\begin{assumption}\label{ass:target}
The \emph{target} $g$ lies in $\Ltwo$. For branch instances with
$\mu=P$ a probability law on $\R$ and polynomial dictionaries of degree $s$,
this amounts to finite moments of $P$ up to order $2s$; for the detection branch
it is the $\chi^2$-integrability condition of Table~\ref{tab:dictionary}.
\end{assumption}

\begin{assumption}\label{ass:gram}
The Gram matrix $\F=(\ipmu{\varphi_i}{\varphi_j})_{i,j=0}^{S}$ may be singular;
$\F^{+}$ denotes the Moore--Penrose pseudoinverse, and $\F^{-1}$ is used when
$\dict_S$ is linearly independent.
\end{assumption}

Write $\spanH(\dict_S)=\operatorname{span}\dict_S$,
$\Y=(\ipmu{\varphi_i}{g})_{i=0}^{S}$, and let $\proj{\spanH}$ denote orthogonal
projection onto $\spanH$.

\subsection{The normal system is a projection}

\begin{theorem}[normal system $\Leftrightarrow$ projection]\label{thm:normal-system}
Let Assumptions~\ref{ass:dict}--\ref{ass:gram} hold and let
$\K\in\mathbb{K}^{S+1}$. Then
\[
  \F\K=\Y
  \quad\Longleftrightarrow\quad
  \sum_{i=0}^{S}K_i\varphi_i=\proj{\spanH(\dict_S)}\,g .
\]
Solutions $\K$ exist for every $g$; $\K$ is unique if and only if $\F$ is
nonsingular. For any solution,
\begin{equation}\label{eq:energy}
  J(\dict_S;g)\;:=\;\bigl\lVert\proj{\spanH(\dict_S)}\,g\bigr\rVert^{2}
  \;=\;\K^{*}\F\K\;=\;\K^{*}\Y\;=\;\Y^{*}\F^{+}\Y .
\end{equation}
\end{theorem}

\begin{proof}
The $j$-th equation of $\F\K=\Y$ reads
$\ipmu{\varphi_j}{\,g-\sum_iK_i\varphi_i}=0$. If all $S+1$ equations hold, the
residual is orthogonal to every generator, hence to $\spanH(\dict_S)$ by
linearity, and $\sum_iK_i\varphi_i\in\spanH(\dict_S)$ is the unique element of
the span with this property---the orthogonal projection. Conversely the residual
of the projection is orthogonal to each $\varphi_j$. Existence follows because
the projection lies in the span; uniqueness of the \emph{coefficients} is linear independence, i.e.\ nonsingularity of
$\F$ (the projection itself is unique regardless: when $\F$ is singular, every
solution $\K$ represents the same element of the span). For \eqref{eq:energy}, write
$v=\proj{\spanH}g$; then $\ipmu{v}{g-v}=0$ gives
$\norm{v}^2=\ipmu{v}{g}=\sum_i \bar K_i\ipmu{\varphi_i}{g}=\K^{*}\Y$, and
$\K^{*}\Y=\K^{*}\F\K$ from the normal system; the pseudoinverse form follows by
restricting to a maximal independent subfamily.
\end{proof}

Theorem~\ref{thm:normal-system} is elementary, and we state it for two reasons.
First, it is the form all five branches of Table~\ref{tab:dictionary} actually
use: a non-orthogonal dictionary, a possibly ill-conditioned Gram matrix, and a
functional of the projection. Second, this exact statement---general measure,
no independence assumption---is what Appendix~\ref{app:lean} machine-checks,
closing an open question of \citet{zabolotnii2026gsa}, where only the
orthonormal Parseval form was formalized.

\begin{proposition}[monotonicity, saturation, exhaustion, redundancy]\label{prop:monotone}
Let $\spanH\subseteq\spanH'$ be finite-dimensional subspaces of $\Ltwo$ and
$g\in\Ltwo$. Then
(a) $J(\spanH;g)\le J(\spanH';g)$;
(b) $J(\spanH;g)\le\normmu{g}^{2}$, with equality iff $g\in\spanH$;
(c) if $\spanH_1\subseteq\spanH_2\subseteq\cdots$ with
$\overline{\bigcup_s\spanH_s}\ni g$, then $J(\spanH_s;g)\uparrow\normmu{g}^{2}$;
(d) $J(\spanH;g)=J(\spanH';g)$ iff $\proj{\spanH'}g\in\spanH$.
\end{proposition}

\begin{proof}
(a) follows from $\proj{\spanH}=\proj{\spanH}\proj{\spanH'}$ and
$\norm{\proj{\spanH}}\le1$; (b) and (c) are the Pythagoras and density
arguments; (d) restates equality in (a).
\end{proof}

Proposition~\ref{prop:monotone} consolidates known facts. Its GMM
instance---adding moment conditions never lowers efficiency, with (d) as the
redundancy characterization---is \citet{breusch1999redundancy}; its ECF instance
under frequency-grid refinement is \citet{feuerverger1981efficiency}; its
detection instance is Theorem~2 of \citet{zabolotnii2026gsa}. We record it
because every row of Table~\ref{tab:dictionary} invokes some part of
(a)--(d), and because the redundancy case (d) appears not to have been stated in
generating-element terms before.

\subsection{The five-branch dictionary}\label{sec:dictionary}

Table~\ref{tab:dictionary} instantiates the pair
(Theorem~\ref{thm:normal-system}, Proposition~\ref{prop:monotone}) across the
five branches. Each row lists the base measure, the dictionary, the projected
target, the functional of the projection the branch optimizes or reports, and
the published result the instantiation recovers. Each recovery is a theoretical
consequence of Theorem~\ref{thm:normal-system} under the row's measure and target;
Section~\ref{sec:verification} confirms them numerically to near machine precision
as an independent check, not as the basis of the claim.

The GSA and DSGE rows each need a caveat. The GSA row requires
$\ell=f_1/f_0-1\in L^{2}(P_0)$---finiteness of the $\chi^{2}$ divergence, which is
easy to violate: an alternative $f_1$ with heavier tails than $f_0$ makes
$\chi^2(f_1\Vert f_0)=\infty$ while every entry of $\F$ and $\Y$ stays finite. Our
own verification initially used such an alternative until a failed saturation check
exposed it; we flag the condition as a standing hypothesis the branch literature
leaves implicit. The DSGE row uses a regularized class measure
($\F_c+\lambda\Imat$); its residual identity holds exactly on the empirical
measure, which is how we verify it.

\begin{sidewaystable}
\centering
\caption{The five-branch dictionary (PMM, polynomial maximization method; GSA,
generalized stochastic approximation; DSGE, decomposition in a space with a
generating element; CF, characteristic function; VWK, Volterra--Wiener--Kunchenko).
Each branch is the projection of Theorem~\ref{thm:normal-system} under its own
base measure and target; $\cum{k}$ are standardized cumulants, $h_k$ monic Hermite
polynomials, $\psi_k$ the $P$-matched orthonormal polynomials. The last column
names the published result that the row recovers when the projection of
Theorem~\ref{thm:normal-system} is instantiated with that measure and target.}
\label{tab:dictionary}
\footnotesize
\setlength{\tabcolsep}{4pt}
\begin{tabular}{@{}llp{2.6cm}p{4.2cm}p{5.4cm}@{}}
\toprule
Branch & measure $\mu$ & target $g$ & functional & recovers \\
\midrule
PMM (estimation) & $P$, centered & $\varphi_1=x-m_1$ &
  $g_S=1-J_{\mathrm{aug}}/\normP{\varphi_1}^2$ &
  $g_2=1-\cum{3}^2/(2+\cum{4})$ \citep{kunchenko2002polynomial} \\
GSA (detection) & $P_0$ pre-change & $\ell=f_1/f_0-1\in L^2(P_0)$ &
  $J(s)=\Y^{\top}\F^{-1}\Y$ &
  Thm.~2 of \citet{zabolotnii2026gsa} \\
DSGE (classification) & per-class $P_c$ & observation $x$ &
  $\log$-MSED $=\log(\normP{g}^2-J)$ &
  residual energy $\normP{g}^2-J$ of~\eqref{eq:energy} \\
CF (min.\ distance) & $\sum_m w_m\delta_{u_m}$, $\mathbb{C}$ &
  score combination $a^{\top}s_\theta$ &
  $a^{\top}J(U)a$ &
  grid monotonicity \citep{feuerverger1981efficiency,zabolotnii2026momentfree} \\
VWK (identification) & input law $P$ & response $f\in\Pi_s$ &
  matched coefficients; mismatch $\to$ $\pen(s)$ &
  Thm.~1 and the $s=2$ penalty of \citet{zabolotnii2026vwk} \\
\bottomrule
\end{tabular}
\end{sidewaystable}

\section{Cross-branch identities and a computable diagnostic}\label{sec:identities}

\subsection{The penalty and PMM variance reduction share their cumulant ratios}

The variance-reduction coefficient of order-$S$ PMM estimation measures the gain
from augmenting the linear estimating function with higher polynomial terms:
$g_S=1-J_{\mathrm{aug}}/\normP{\varphi_1}^{2}$ in the notation of
Table~\ref{tab:dictionary}, with the classical values
$g_2=1-\cum{3}^{2}/(2+\cum{4})$ and
$g_3=1-\cum{4}^{2}/(6+9\cum{4}+\cum{6})$
\citep{kunchenko2002polynomial,zabolotnii2018polynomial}. Theorem
\ref{thm:closed-forms}(i) exhibits $1-g_3$ as the symmetric order-3 penalty per
unit of odd-signal energy. The order-2 analogue holds as well: specializing the
$s=2$ closed form to signals with $b_2=0$ gives
\begin{equation}\label{eq:g2-identity}
  \pen(2;(0,b_1,0),P)\;=\;\frac{\cum{3}^{2}}{2+\cum{4}}\;b_1^{2}
  \;=\;(1-g_2)\,b_1^{2}.
\end{equation}
Both identities are verified to $2\cdot10^{-16}$ in
Section~\ref{sec:verification}.

The interpretation is symmetric in the two branches. PMM \emph{gains} the factor
$1-g_S$ in variance by exploiting the non-Gaussian structure of the error law;
the diagonal Wiener rule \emph{loses} the same factor in excess risk by ignoring
the non-Gaussian structure of the input law. Both quantities are squared norms
in the same projection geometry---the PMM gain is the energy of the score
component captured by the augmenting dictionary, the identification penalty is
the energy of the signal component mis-handled by the mismatched dictionary---and
the cumulant ratios $\cum{3}^{2}/(2+\cum{4})$, $\cum{4}^{2}/(6+9\cum{4}+\cum{6})$
are the exchange rates. We are not aware of a prior quantitative link between
these two literatures.

\subsection{The characteristic-function branch realizes
\texorpdfstring{$g_2$}{g2} at narrow band}\label{sec:cf-g2}

The same cumulant ratio ties in a third branch. The CF row of
Table~\ref{tab:dictionary} optimizes the optimal-weight information
$a^{\top}J(U)a$ over a frequency grid $U=\{u_1,\dots,u_M\}$ and recovers grid
monotonicity \citep{feuerverger1981efficiency,carrasco2000generalization}:
refining $U$ raises $J(U)$ to the Fisher information, the exhaustion case (c) of
Proposition~\ref{prop:monotone}. Its opposite endpoint is the PMM coefficient
$g_2$, so the estimation and characteristic-function rows meet at a single
number.

Take location estimation $x=\theta+\xi$ under a centered law $P$ with variance
$\sigma^{2}$ and standardized cumulants $\cum{3},\cum{4}$, and the one-frequency
real moment pair $\varphi^{\cos}_u=\cos u(x-\theta)$,
$\varphi^{\sin}_u=\sin u(x-\theta)$ of the CF branch. The optimally weighted
minimum-CF-distance estimator on the single frequency $u$ has information
$J(u)=\Gmat^{\top}\Omega^{-1}\Gmat$ and variance $V(u)=J(u)^{-1}$---the order-2
GMM/minimum-distance form \citep{hansen1982large,feuerverger1981efficiency}---with
$\Gmat=\E\,\partial_\theta(\varphi^{\cos}_u,\varphi^{\sin}_u)^{\top}$ and
$\Omega=\operatorname{Cov}_P(\varphi^{\cos}_u,\varphi^{\sin}_u)$.

\begin{proposition}[narrow-band realization of $g_2$]\label{prop:cf-g2}
With the notation above,
\[
  \frac{V(u)}{\sigma^{2}}\;\xrightarrow[u\to0]{}\;
  g_2=1-\frac{\cum{3}^{2}}{2+\cum{4}},
\]
and grid refinement raises the information monotonically from
$g_2^{-1}\sigma^{-2}$ to the Fisher information by
Proposition~\ref{prop:monotone}(a,c). The characteristic-function branch
therefore interpolates between order-2 PMM efficiency and full efficiency, with
$g_2$ as its narrow-band coefficient.
\end{proposition}

\begin{proof}
Using $\E\xi=0$, $\E\xi^{2}=\sigma^{2}$, $\E\xi^{3}=\cum{3}\sigma^{3}$,
$\E\xi^{4}=(\cum{4}+3)\sigma^{4}$ and Taylor expansion in $u$,
$\Gmat=\bigl(-u^{4}\cum{3}\sigma^{3}/6,\,-u\bigr)^{\top}+o(\cdot)$ and
\[
  \Omega=
  \begin{pmatrix}
    \tfrac{u^{4}}{4}(\cum{4}+2)\sigma^{4} & -\tfrac{u^{3}}{2}\cum{3}\sigma^{3}\\[3pt]
    -\tfrac{u^{3}}{2}\cum{3}\sigma^{3} & u^{2}\sigma^{2}
  \end{pmatrix}+o(\cdot),
\]
whence $\det\Omega=\tfrac{u^{6}}{4}\sigma^{6}\bigl[(\cum{4}+2)-\cum{3}^{2}\bigr]+o(u^{6})$
and $\Gmat^{\top}\Omega^{-1}\Gmat\to(\cum{4}+2)\big/\bigl\{\sigma^{2}[(\cum{4}+2)-\cum{3}^{2}]\bigr\}$;
inverting gives $V(u)/\sigma^{2}\to1-\cum{3}^{2}/(2+\cum{4})=g_2$. The monotone
rise to the Fisher information is Proposition~\ref{prop:monotone}(a,c) applied to
the nested grids \citep{feuerverger1981efficiency,carrasco2000generalization}.
\end{proof}

\begin{remark}\label{rem:cf-g2}
Proposition~\ref{prop:cf-g2} links two endpoints the calculus already
provides; it is not a separate efficiency result. Optimally weighted GMM is
invariant under smooth invertible reparametrization of the moment span
\citep{hansen1982large}, and $\{\cos u\xi,\sin u\xi\}$ spans
$\{\xi,\xi^{2}-\sigma^{2}\}$ to leading order, so the one-frequency CF estimator
and order-2 PMM share their estimating space---and hence $g_2$---exactly. Past
the leading term $V(u)/\sigma^{2}=g_2+O(u^{2})$, with corrections involving
$\cum{5},\cum{6},\dots$ and the grid; the efficiency of a fixed-bandwidth
characteristic kernel is thus a functional of the whole law, not a closed
function of $(\cum{3},\cum{4})$. The CF row sits on the estimation
(score) side of the $\LtwoP$ dichotomy: it matches the score functional, not the
signal target of the VWK and Wiener rows.
\end{remark}

\subsection{A plug-in diagnostic}\label{sec:diagnostic}

Theorem~\ref{thm:penalty} turns basis selection into a computation. Given data,
estimate the moments of the input law to order $2s$, plug them into
\eqref{eq:penalty-matrix}, and read off $\widehat{\pen}(s)$: if it is small
relative to the signal energy, the tabulated Gaussian basis is adequate at order
$s$; otherwise, a matched basis is needed. The diagnostic costs
$O(s^{3})$ operations after a single pass over the data.

When the signal $\bb$ is not yet available---the basis-selection decision often
precedes any fit---the signal-free certificate of Proposition~\ref{prop:bound}
applies: it bounds the penalty per unit signal energy by
$\tfrac{1+\varepsilon}{1-\varepsilon}\varepsilon^{2}$ from the off-diagonal mass
$\varepsilon$ of the Wiener-family Gram correlations alone, so a small
$\varepsilon$ certifies adequacy without $\bb$. The signal-aware penalty
\eqref{eq:penalty-matrix} then refines the verdict once $\bb$ is estimated
(for instance from a preliminary full-Gram solve on the same data).

We assessed the plug-in's finite-sample behaviour by Monte Carlo: 500
replications of samples of size $n\in\{500,2000,8000\}$ from three laws
(centered Gamma(3); uniform; the contaminated normal
$0.9\,N(0,0.25)+0.1\,N(0,4)$), comparing $\widehat{\pen}(s)$ against the
population value for $s\in\{2,3\}$. Table~\ref{tab:diagnostic} reports the
fraction of replications within $10\%$ of the population penalty and the median
relative error. At $n=2000$, $s=3$, the diagnostic is reliable for the
moderate-tail laws (83\% within $10\%$ for Gamma(3), 99\% for uniform; median
relative errors $0.048$ and $0.030$) and improves further at $n=8000$. For the
contaminated law it is not (12\% within $10\%$; median relative error $0.363$):
$\widehat{\pen}(3)$ depends on sixth-order empirical moments whose sampling
variance is driven by the $\sigma=2$ component. We therefore scope the claim
explicitly: the plug-in diagnostic is a practical tool at $n\approx2000$ for
moderate-tail inputs, and a population-level statement under heavy
contamination. Since the failure is a variance problem and not a bias one, the
remedy is to quantify the uncertainty rather than to ``robustify'' the point
estimate; we calibrate a bootstrap below.

\begin{table}[t]
\centering
\caption{Plug-in diagnostic $\widehat{\pen}(s)$ vs.\ the population penalty:
fraction of 500 replications within $10\%$ (med.\ rel.\ error in parentheses).
At $\pen=0$ cells (symmetric laws at $s=2$) the criterion is
$\widehat{\pen}\le0.02$.}
\label{tab:diagnostic}
\small
\begin{tabular}{@{}llccc@{}}
\toprule
Law & $s$ & $n=500$ & $n=2000$ & $n=8000$ \\
\midrule
centered Gamma(3) & 2 & 0.61\,(0.073) & 0.90\,(0.037) & 1.00\,(0.021) \\
                  & 3 & 0.44\,(0.114) & 0.83\,(0.048) & 0.99\,(0.032) \\
uniform           & 2 & 0.95\,(0.003) & 1.00\,(0.001) & 1.00\,(0.000) \\
                  & 3 & 0.75\,(0.059) & 0.99\,(0.030) & 1.00\,(0.014) \\
contaminated      & 2 & 0.49\,(0.020) & 0.79\,(0.007) & 0.99\,(0.002) \\
                  & 3 & 0.11\,(0.484) & 0.12\,(0.363) & 0.20\,(0.247) \\
\bottomrule
\end{tabular}
\end{table}

\begin{proposition}[sampling distribution of the plug-in]\label{prop:plug-in-clt}
Fix $s$ and a signal $\bb$, and let $\widehat{\pen}(s)$ be
\eqref{eq:penalty-matrix} evaluated at the empirical moments
$\widehat m_1,\dots,\widehat m_{2s}$. On the open set where the Hankel matrix
$\Hmat_s$ is positive definite, $\pen(s;\bb,P)$ is a smooth ($C^{\infty}$)
function of those moments---the Cholesky factorization and matrix inversion that
define it are smooth on the positive-definite cone---so if $P$ has finite moments
to order $4s$ the empirical-moment central limit theorem and the delta method
give
\[
  \sqrt{n}\,\bigl(\widehat{\pen}(s)-\pen(s;\bb,P)\bigr)
  \;\xrightarrow{d}\;N\!\bigl(0,\;\nabla\pen^{\top}\boldsymbol{\Sigma}\,\nabla\pen\bigr),
\]
where $\boldsymbol{\Sigma}$ is the asymptotic covariance of
$(\widehat m_1,\dots,\widehat m_{2s})$, with $\boldsymbol{\Sigma}_{ij}=m_{i+j}-m_im_j$,
and $\nabla\pen$ is the gradient of the penalty in the moments. The order $4s$ is
necessary, since $\boldsymbol{\Sigma}$ already involves $m_{4s}$. A studentized
interval follows by plugging empirical moments into
$\nabla\pen^{\top}\boldsymbol{\Sigma}\,\nabla\pen$; at heavy tails, where that
plug-in is itself unstable, the nonparametric bootstrap is the safer route and is
the one we calibrate below.
\end{proposition}

\paragraph{Choosing the order.}
The same sampling theory yields a stopping rule. By the mechanism of
Remark~\ref{rem:mechanism} each order recruits the next cumulant the Gaussian
basis fails to annihilate, so the penalty grows by discrete, cumulant-driven
increments. One computes $\widehat{\pen}(s)$ for increasing $s$ and stops at the
first order whose increment $\widehat{\pen}(s)-\widehat{\pen}(s-1)$ is within a
few standard errors of zero (Proposition~\ref{prop:plug-in-clt}): beyond it the
cumulant the next order would activate is not resolvable at the available sample
size, and matching the basis cannot be justified from the data. This is the
mismatch-penalty analogue of the $J(s)$-saturation test that the detection branch
performs through Proposition~\ref{prop:monotone}.

\paragraph{Bootstrap calibration.}
Proposition~\ref{prop:plug-in-clt} is borne out in finite samples. Across the
three laws at $n=2000$, $s=3$, a nonparametric bootstrap ($200$ resamples) yields
well-calibrated intervals---empirical coverage of the population penalty
$0.86$--$0.90$ against the nominal $0.90$---and an interval width that separates
stable and unstable cases: the $90\%$ interval half-width is about $7\%$ of the
penalty for the uniform law but exceeds $100\%$ of it for the contaminated law,
where the practitioner is thereby warned not to trust the point value. A
trimmed or winsorized cumulant estimator is, by contrast, \emph{not}
a fix: by lowering variance it targets the bulk of the law and badly
underestimates the penalty of a heavy-tailed input---for centered Gamma$(3)$ a
$5\%$ winsorization returns essentially zero against a population value of $1.54$.
Here the useful safeguard is the uncertainty interval supplied by the bootstrap;
the calibration run is in the reproducibility bundle.

\paragraph{A real measurement chain.}
A field illustration comes from the Wiener--Hammerstein benchmark
\citep{schoukens2009wienerhammerstein}, an electronic circuit
recorded at $100{,}000$ samples. Treating each measured signal's marginal law as
the input law of a cross-correlation identification, the diagnostic separates two
regimes drawn from the \emph{same} chain and returns opposite, actionable
basis-selection verdicts. The designed excitation is near-Gaussian
(skewness $-0.01$): $\widehat{\pen}(2)=0.0001$ and $\widehat{\pen}(3)=0.0013$, with
$90\%$ block-bootstrap intervals close to zero ($[0,0.0005]$ and $[0.0001,0.0049]$)
---the Gaussian basis is adequate for this signal. The circuit's output, skewed
by the static nonlinearity (skewness $-0.31$), has nonzero penalty:
$\widehat{\pen}(2)=0.050$
($[0.039,0.060]$) and $\widehat{\pen}(3)=0.206$ ($[0.187,0.224]$), intervals well
clear of zero, and the nonzero penalty already at $s=2$ is the skewness signature
of Remark~\ref{rem:mechanism}. A moving-block bootstrap (length $1000$) accounts
for the serial dependence of the records.

\section{Numerical verification}\label{sec:verification}

Our claims are algebraic identities and exact formulas, so the verification
burden is different from that of an empirical method paper: what must be checked
is that the formulas, the symbolic derivations, and the recoveries of
Table~\ref{tab:dictionary} agree with computations that share none of their
machinery. We use three independent paths and report every disagreement.

\subsection{Protocol}

\emph{Path A (moment algebra)} evaluates \eqref{eq:penalty-matrix} and the
dictionary functionals through Hankel--Cholesky factorizations in double
precision. \emph{Path B (adaptive quadrature)} integrates the defining
expectations directly against the density of each law. \emph{Path C (randomized
quasi--Monte Carlo)} estimates the same expectations by scrambled Sobol' points
through inverse transforms, eight independent scramblings providing a standard
error; a plain Monte Carlo path with three seeds adds a $z$-consistency check.
The claim grid covers six laws (Gaussian; centered exponential; centered
Gamma(3); uniform; the contaminated normal $0.9N(0,0.25)+0.1N(0,4)$; a centered
lognormal stress row), orders $s\in\{2,\dots,5\}$, and one canonical plus three
random coefficient vectors per cell---96 cells per seed.

\subsection{Results}

Path A agrees with Path B on \emph{all} 96 cells to relative gaps below
$10^{-6}$; on the canonical cells of Table~\ref{tab:verification} the gaps are
at the $10^{-15}$ level, i.e.\ machine precision. Path C reproduces at least
three significant digits on the core grid ($s\in\{2,3,4\}$, five core laws),
and the order-threshold and cross-branch identities of
Sections~\ref{sec:penalty}--\ref{sec:identities} hold to $2\cdot10^{-16}$.
The five recoveries of Table~\ref{tab:dictionary} hold to deviations below
$1.4\cdot10^{-14}$.

\begin{table}[t]
\centering
\caption{Verification on the canonical coefficient vectors: penalty value,
relative gap between the matrix formula (Path A) and adaptive quadrature
(Path B), significant digits of agreement with scrambled-Sobol' QMC (Path C),
and condition number of the Hankel matrix. ``$0$ (exact)'' marks cells where the
penalty vanishes identically (Corollary~\ref{cor:threshold}) and both paths
return values below $10^{-30}$.}
\label{tab:verification}
\small
\begin{tabular}{@{}lccccc@{}}
\toprule
Law & $s$ & $\pen(s)$ & $|$A$-$B$|/\pen$ & QMC digits & $\cond(\Hmat_s)$ \\
\midrule
Gaussian            & 2--5 & $0$ (exact) & --- & exact & $5.8\,$--$\,4.9\cdot10^{3}$ \\
centered Exp(1)     & 3 & 2.0593 & $2\cdot10^{-16}$ & 3.7 & $9.9\cdot10^{2}$ \\
                    & 4 & 4.7544 & $8\cdot10^{-16}$ & 5.0 & $7.3\cdot10^{4}$ \\
                    & 5 & 8.1631 & $1\cdot10^{-15}$ & 3.6 & $8.7\cdot10^{6}$ \\
centered Gamma(3)   & 3 & 1.5445 & $4\cdot10^{-16}$ & 3.7 & $3.7\cdot10^{3}$ \\
                    & 4 & 3.8514 & $5\cdot10^{-16}$ & 5.3 & $3.0\cdot10^{5}$ \\
                    & 5 & 6.9088 & $4\cdot10^{-16}$ & 4.1 & $3.7\cdot10^{7}$ \\
uniform             & 2 & $0$ (exact) & --- & exact & $7.7$ \\
                    & 3 & 0.5600 & $2\cdot10^{-16}$ & 15.1 & $3.6\cdot10^{1}$ \\
                    & 4 & 1.0286 & $2\cdot10^{-15}$ & 9.7 & $2.1\cdot10^{2}$ \\
contaminated        & 2 & $0$ (exact) & --- & exact & $8.1$ \\
                    & 3 & 0.2313 & $<10^{-16}$ & 4.3 & $2.6\cdot10^{2}$ \\
                    & 4 & 0.5253 & $9\cdot10^{-16}$ & 4.0 & $7.3\cdot10^{3}$ \\
centered logN$(0,\!.5^2)$ & 3 & 1.6882 & $5\cdot10^{-16}$ & 3.1 & $1.6\cdot10^{2}$ \\
                    & 4 & 3.6870 & $2\cdot10^{-15}$ & 2.6 & $9.8\cdot10^{3}$ \\
                    & 5 & 5.7889 & $1\cdot10^{-15}$ & 1.7 & $1.6\cdot10^{6}$ \\
\bottomrule
\end{tabular}
\end{table}

Beyond the table, the protocol produced three findings of independent interest.

\paragraph{Conditioning binds less tightly than condition numbers suggest.}
A multiprecision-referenced sweep across law families and orders
$s\in\{2,\dots,8\}$ (full settings in Appendix~\ref{app:repro}) shows the
double-precision Cholesky pipeline keeping six significant digits through $s=8$
even at $\cond(\Hmat_8)\approx10^{21}$. The projection functionals are
residual-type quantities computed through backward-stable solves, so their
forward error tracks an effective conditioning far smaller than $\cond(\Hmat)$;
the condition-number alarms calibrated for coefficient recovery (the rule-of-thumb
$\cond(\Hmat)\gtrsim10^{6}$ warning common in numerical practice) are, in the
tested families and orders, too pessimistic for penalty and energy evaluation. In
the same sweep ridge regularization of population Hankel matrices was not
helpful---it biases the functional and, in particular, destroys the exact zeros of
Corollary~\ref{cor:threshold}---so regularization belongs to the
\emph{empirical}-moment regime of Section~\ref{sec:diagnostic}, not to the
population algebra.

\paragraph{Sampling-based verification has a hard floor on heavy tails.}
On the lognormal stress row at $s=5$ the QMC path stalls below two digits and on
one random signal lands $12$ standard errors away from Paths A--B, which agree
with each other to twelve digits there. The discrepancy is fully accounted for
by tail truncation: scrambled Sobol' points at $2^{25}$ resolution cannot reach
beyond $u^{*}=1-2^{-25}$ of the uniform scale, and direct quadrature of the
integrand mass beyond the corresponding quantile gives $T=0.9508$, matching the
QMC deficit within $4$ standard errors. All scramblings share the truncation, so
the internal error estimate cannot see it. This is a deceptive failure mode of
randomized QMC on unbounded heavy-tailed integrands of high polynomial degree,
and one more argument for exact moment formulas where they exist.

\paragraph{The penalty is the realizable performance gap.}
The penalty is a population functional, but it also governs finite-sample
identification. We simulated noisy data ($Y=f(X)+\varepsilon$,
$\varepsilon\sim N(0,0.3^{2})$, the canonical cubic signal) and estimated the
coefficients two ways---the diagonal Wiener rule \eqref{eq:diagonal-rule} and the
matched non-diagonal rule (a full least-squares solve in the matched
basis)---measuring each reconstruction's excess $\LtwoP$ risk against the true
response. Averaged over $300$ replications, the matched rule's risk decays at the
parametric rate ($5\cdot10^{-4}\to5\cdot10^{-5}$ as $n:500\to8000$), while the
diagonal rule's risk converges to the population penalty: to $\pen(3)=1.54$ for
centered Gamma(3) ($1.43,1.52,1.56$ at $n=500,2000,8000$) and to $\pen(3)=0.23$
for the contaminated normal ($0.62,0.35,0.25$), the heavier law approaching its
floor more slowly. The gap between the two strategies converges to the
penalty as $n$ grows, exactly as the population formula
\eqref{eq:penalty-matrix} predicts---the irreducible cost the diagonal rule
pays and the matched rule removes, not an artifact of the finite-sample
estimator.

\subsection{One object, five functionals}

Figure~\ref{fig:five-functionals} renders the dictionary of
Table~\ref{tab:dictionary} on a single object: the centered Gamma(3) law with
the power dictionary. One shared computational core (moments, Hankel--Cholesky,
projection energies) produces all five panels: the PMM variance-reduction curve
$g_S$ decreasing in $S$; the detection information $J(s)$ increasing to its
$\chi^{2}$ saturation; the classification log-MSED matrix separating two classes
by $6.7$ nats; the ECF information $J(U)$ increasing to the Fisher information
under grid refinement; and the identification penalty $\pen(s)$ growing with
order. No branch-specific derivation enters the code path---only the base
measure and the target change across panels, which is the dictionary made
visible.

\begin{figure}[t]
\centering
\includegraphics[width=\linewidth]{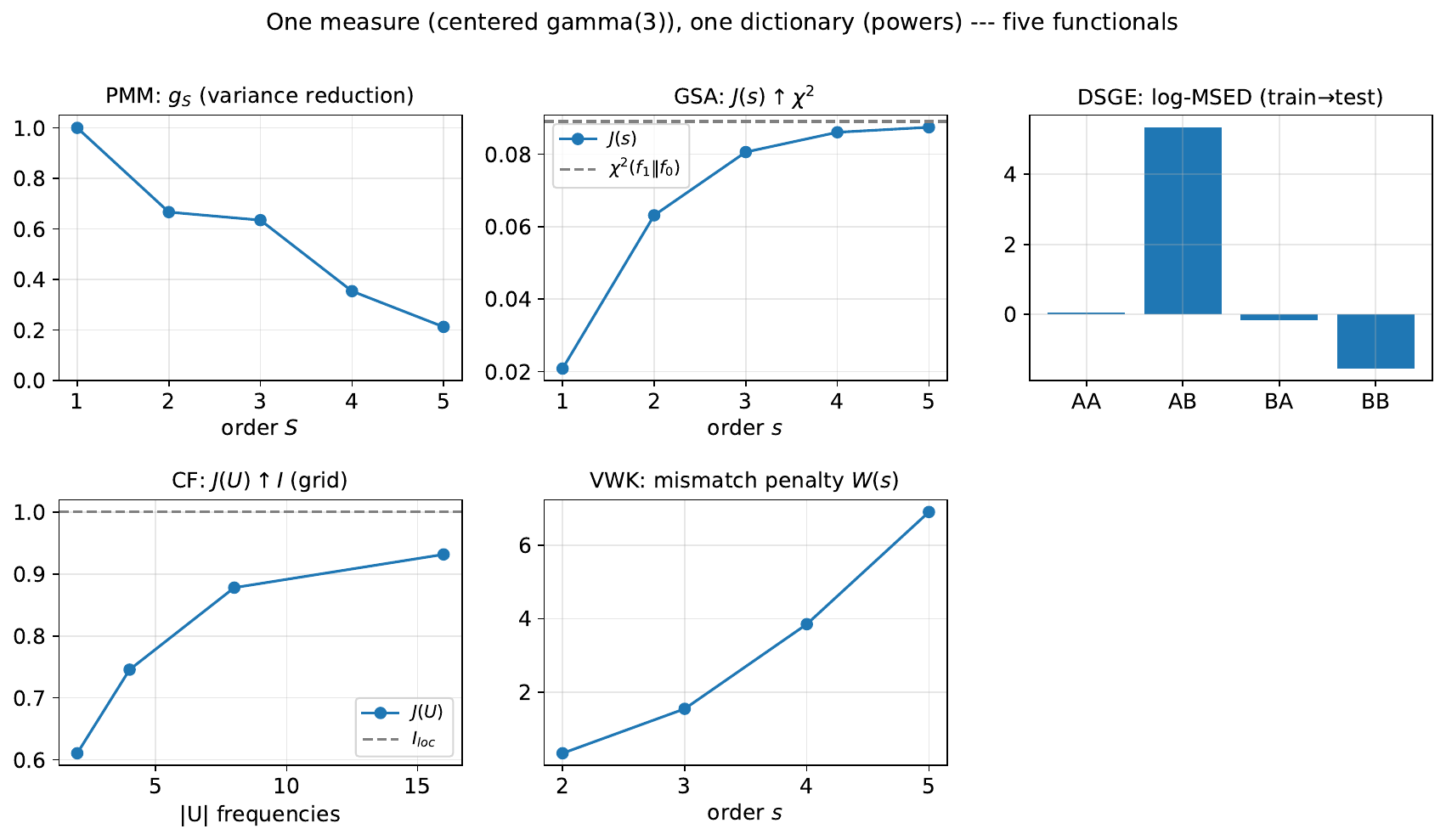}
\caption{One projection core reproduces all five dictionary functionals on a
single law (centered Gamma(3), power dictionary). Top row: PMM (polynomial
maximization method) variance reduction $g_S$; GSA (generalized stochastic
approximation) information $J(s)$ with its $\chi^{2}$ limit; DSGE
(decomposition in a space with a generating element) cross-class log-MSED,
the log of the reconstruction residual energy $\normP{g}^2-J$. Bottom row: CF
(characteristic function) information $J(U)$ under frequency-grid refinement
with the Fisher limit; VWK (Volterra--Wiener--Kunchenko) mismatch penalty
$\pen(s)$. Abbreviations as in Table~\ref{tab:dictionary}.}
\label{fig:five-functionals}
\end{figure}

\section{Discussion}\label{sec:discussion}

\paragraph{Division of labor with companion work.}
The contribution boundary drawn in Section~\ref{sec:intro} bears repeating only
where it touches ongoing work: the companion manuscript on
distribution-matched Volterra bases \citep{zabolotnii2026vwk} establishes the
$s=2$ penalty and the orthogonalization construction, and the PMM--least-squares
bridge developed elsewhere establishes the $S=2$ equivalence with second-order
least squares \citep{wang2008asymptotic} together with the $S\ge3$ efficiency
reserve. This paper supplies the umbrella both instantiate, plus the
general-$s$ results that neither contains.

\paragraph{Limitations.}
Four limitations remain. First, the study is primarily theoretical-numerical:
the real-data illustration of Section~\ref{sec:diagnostic} is a single
measurement chain (the Wiener--Hammerstein benchmark), and broader field
validation across operational chains remains to be done.
Second, the plug-in diagnostic has high variance under heavy contamination
(Table~\ref{tab:diagnostic}); the bootstrap of Section~\ref{sec:diagnostic} reports
this as wide intervals rather than a misleading point value, but a
distribution-free \emph{point} estimator that is both low-variance and unbiased
under contamination remains out of reach (trimming buys variance at the cost of
bias). Third, the
general skewed closed form at $s=3$ does not factor compactly; the matrix form
is the recommended computation; we report this directly instead of simplifying
an expression of little additional insight. Fourth, all dictionaries here are univariate; the multivariate case follows a
Kronecker-product structure but its cumulant bookkeeping is left to future
work.

\paragraph{Extensions.}
The multivariate case is structurally transparent for a \emph{product} input
law: the matched and Wiener families are tensor products of their univariate
factors, so $\Cmat=\bigotimes_d\Cmat^{(d)}$ and $\Dmat=\bigotimes_d\Dmat^{(d)}$
are Kronecker products of the per-coordinate factors, and
Theorem~\ref{thm:penalty} applies with these factored ingredients. The residual
operator is then $\Imat-\bigotimes_d\Cmat^{(d)\top}\Dmat^{(d)-1}\Cmat^{(d)}$,
which does \emph{not} collapse to a sum of per-coordinate residuals: the penalty
remains a single quadratic form in the coefficient tensor $\bb$ and reduces to a
per-coordinate product only for a rank-one (separable) signal. New
phenomena appear only when the input coordinates are dependent: the Gram
no longer factors, and \emph{mixed} cumulants (co-skewness, co-kurtosis) enter as
cross-moment defects with no univariate analogue. A vector input whose every
marginal is symmetric can then still incur a penalty at $s=3$ through a nonzero
co-kurtosis, so the order-threshold corollary must be restated in terms of the
joint cumulant structure rather than marginal kurtosis. Tensor-moment bookkeeping
and these mixed-cumulant forms deserve their own treatment, and we have kept them
out of scope. On the formal side, extending the Lean module from the
correspondence and monotonicity layer to the closed cumulant forms of
Theorem~\ref{thm:closed-forms} would mechanize the paper's main results
end-to-end. On the methodological side, the diagnostic of
Section~\ref{sec:diagnostic} invites a sequential version: monitor
$\widehat{\pen}(s)$ online and switch bases when the estimated penalty crosses a
cost threshold---a design problem the detection branch of the dictionary is
already equipped to analyze.

\paragraph{Take-away.}
A practitioner who knows the first $2s$ moments of the input law can now write
down, exactly and in $O(s^{3})$ operations, the price of keeping the Gaussian
basis at order $s$. The same projection geometry also explains why this price
uses the cumulant ratios that appear in the classical efficiency gains of
polynomial estimation.

\appendix
\setcounter{table}{0}
\setcounter{figure}{0}
\section{Proofs and derivations}\label{app:proofs}

\subsection{Proposition~\ref{prop:monotone}(d): the redundancy case}
Suppose $J(\spanH;g)=J(\spanH';g)$ with $\spanH\subseteq\spanH'$. Writing
$v=\proj{\spanH}g$ and $v'=\proj{\spanH'}g$, the tower property
$\proj{\spanH}=\proj{\spanH}\proj{\spanH'}$ gives $v=\proj{\spanH}v'$, so
$\norm{v}^{2}=\norm{v'}^{2}$ forces $\norm{v'-v}^{2}=\norm{v'}^{2}-\norm{v}^{2}=0$
by Pythagoras in $\spanH'$, i.e.\ $v'=v\in\spanH$. The converse is immediate.
In Gram form with nested dictionaries this is the
generating-element analogue of the moment-redundancy characterization of
\citet{breusch1999redundancy}.

\subsection{The order-3 odd block: details}\label{app:odd-block}
With $\sigma^{2}=1$ and $P$ symmetric, the matched odd polynomials are
$\psi_1=x$ and $\psi_3=(x^{3}-\mu_4x)/\sqrt{\mu_6-\mu_4^{2}}$, the Wiener odd
polynomials $g_1=x$, $g_3=(x^{3}-3x)/\sqrt6$. The cross-Gram entries are
$q:=\ipP{g_3}{\psi_1}=(\mu_4-3)/\sqrt6=\cum{4}/\sqrt6$ and
$r:=\ipP{g_3}{\psi_3}=\sqrt{(\mu_6-\mu_4^{2})/6}$, with
$d:=\ipP{g_3}{g_3}=(\mu_6-6\mu_4+9)/6=q^{2}+r^{2}$. The odd block of
$\Imat-\Cmat^{\top}\Dmat^{-1}\Cmat$ is
$\frac1d\bigl(\begin{smallmatrix}-q^{2}&-qr\\-qr&q^{2}\end{smallmatrix}\bigr)$
acting on $(b_1,b_3)$, whose image has squared norm
$\frac{q^{2}}{d^{2}}\bigl[(qb_1+rb_3)^{2}+(rb_1-qb_3)^{2}\bigr]
=\frac{q^{2}}{d}(b_1^{2}+b_3^{2})$. Substituting
$\mu_6-6\mu_4+9=6+9\cum{4}+\cum{6}$ yields
Theorem~\ref{thm:closed-forms}(i). The same two-by-two computation at $s=4$ is
unchanged because the odd dictionary $\{x,x^{3}\}$ and the odd Wiener pair do
not change, which proves the odd-block invariance claimed in part (ii).

\subsection{The order-4 even block}\label{app:even-block}
For symmetric $P$ the even matched family is the monic orthogonal triple
$\{u_0,u_2,u_4\}$ and the even Wiener family $\{1,h_2,h_4\}$ with
$h_2=x^{2}-1$, $h_4=x^{4}-6x^{2}+3$. Since $\E_Ph_2=0$ but
$\E_Ph_4=\cum{4}\ne0$ and $\ipP{h_2}{h_4}=8\cum{4}+\cum{6}$, the diagonal rule
contaminates the $h_4$ coefficient with the constant and quadratic signal
components. Carrying out the derivation in monic (rational) coordinates and
normalizing afterwards---a device that keeps every intermediate expression a
rational function of moments---gives the closed form of
Theorem~\ref{thm:closed-forms}(ii); the structured factorization through
$c_0=\E_Ph_4$, $c_2=\ipP{h_2}{h_4}$, $d_2=\normP{h_2}^{2}$ was verified
symbolically (the difference between the raw expansion and the structured form
simplifies to zero in exact arithmetic) and numerically against
\eqref{eq:penalty-matrix} to $1.4\cdot10^{-13}$ on the uniform and contaminated
laws.

\subsection{The general skewed case at $s=3$}\label{app:general-s3}
For $\cum{3}\ne0$ the same monic-coordinate derivation produces
$\pen(3;\cdot,P)$ as a rational function of $(\cum{3},\cum{4},\cum{5},\cum{6})$
and the monic signal coordinates. The expression is exact but does not factor
below roughly $2.3\cdot10^{3}$ elementary operations, and we have found no
structural reading of comparable clarity to the symmetric case. The supplement
(Appendix~\ref{app:repro}) carries the symbolic derivation together with
four-law numerical confirmation at gaps below $1.1\cdot10^{-14}$; the matrix form
\eqref{eq:penalty-matrix} remains the computation of choice.

\section{The machine-checked core}\label{app:lean}

The change-point development of \citet{zabolotnii2026gsa} formalized its
information functional in Lean~4/Mathlib in the orthonormal (Parseval) form and
recorded as an open question the formalization of the general correspondence
between the normal system $\F\K=\Y$ and the $L^{2}$ projection. We close that
question. A Lean~4 module \texttt{GeneratingElementProjection} (Mathlib
v4.26.0) proves, for an arbitrary finite dictionary $\varphi:\mathrm{Fin}\,
(n)\to E$ in a real inner-product space---no orthogonality, independence, or
completeness of $E$ assumed:

\begin{itemize}
\item \texttt{normalSystem\_iff\_starProjection}: $\K$ solves $\F\K=\Y$ if and
  only if $\sum_iK_i\varphi_i$ is the orthogonal projection of the target onto
  the dictionary span (Theorem~\ref{thm:normal-system}(i));
\item \texttt{energy\_eq\_dotProduct}: for any solution,
  $\norm{\proj{\spanH}g}^{2}=\sum_iK_iY_i$ (the $J=\K^{\top}\Y$ form of
  \eqref{eq:energy});
\item \texttt{energy\_mono}, \texttt{energy\_le},
  \texttt{dictSpan\_prefix\_le}: the F-form monotonicity and saturation of
  Proposition~\ref{prop:monotone}(a,b) over dictionary spans, consolidating the
  nested-subspace lemma of the moment-free development
  \citep{zabolotnii2026momentfree};
\item \texttt{gram\_orthonormal}, \texttt{energy\_orthonormal}: the orthonormal
  instance $\F=\Imat$, under which the energy is the Parseval partial
  sum---recovering exactly the form in which $J(s)$ is formalized in
  \citet{zabolotnii2026gsa} and closing the loop between the two developments.
\end{itemize}

The build completes with zero errors, and \texttt{\#print axioms} reports only
the standard Lean axioms (\texttt{propext}, \texttt{Classical.choice},
\texttt{Quot.sound}) for every theorem---no \texttt{sorry}, no custom axioms.
The claim's scope rests on two design points. First, completeness of the
ambient space is not needed: the dictionary span is finite-dimensional, hence
complete, which is precisely the generality the dictionary of
Table~\ref{tab:dictionary} requires (an arbitrary base measure's $L^{2}$).
Second, linear independence is not assumed: the statement quantifies over all
solutions of the normal system, so singular Gram matrices are covered; what
nonsingularity buys is uniqueness, not existence. We present the module as a
certificate of correctness: the paper-level proofs of Section~\ref{sec:calculus}
are elementary, and the Lean development adds machine-checked certainty over
them. Even so, this formalization is, as far as we know, the first of the
normal-system--projection correspondence at this generality; Table~\ref{tab:lean-map}
maps each Lean declaration to its numbered statement in this paper.

\begin{table}[t]
\centering
\caption{Mapping between Lean declarations of the module
\texttt{GeneratingElementProjection} and the statements of this paper. All
declarations compile (Mathlib v4.26.0, 2322 build jobs, zero errors) and audit
to the axioms \{\texttt{propext}, \texttt{Classical.choice},
\texttt{Quot.sound}\}.}
\label{tab:lean-map}
\footnotesize
\begin{tabular}{@{}lp{6.2cm}@{}}
\toprule
Lean declaration & Paper statement \\
\midrule
\texttt{normalSystem\_iff\_starProjection} & Theorem~\ref{thm:normal-system}, equivalence \\
\texttt{energy\_eq\_dotProduct} & Theorem~\ref{thm:normal-system}, eq.~\eqref{eq:energy} \\
\texttt{norm\_starProjection\_mono} & Proposition~\ref{prop:monotone}(a), operator form \\
\texttt{energy\_mono}, \texttt{dictSpan\_prefix\_le} & Proposition~\ref{prop:monotone}(a), dictionary form \\
\texttt{energy\_le} & Proposition~\ref{prop:monotone}(b) \\
\texttt{gram\_orthonormal}, \texttt{energy\_orthonormal} & orthonormal instance; $J(s)$ of \citet{zabolotnii2026gsa} \\
\bottomrule
\end{tabular}
\end{table}

\section{Reproducibility}\label{app:repro}

All code, the Lean module, and the stored numerical outputs are released as a
public MIT-licensed supplement,
\begin{center}
\url{https://github.com/SZabolotnii/Ku-Projection-Framework-code-supplement}
\end{center}
to be archived under a persistent identifier upon acceptance. Implementation
lives entirely there: the paper itself stays code-free, stating only the
population-level recipe (Algorithm~\ref{alg:penalty}). The supplement contains
the shared projection core and every experiment---five-branch specialization
recovery; the conditioning/order sweep against a 50-digit reference; the
three-path penalty verification with seeds; the cross-branch ``one object, five
functionals'' illustration and the plug-in diagnostic; the off-diagonal-bound,
skewed-decomposition and finite-support checks; the bootstrap calibration and the
end-to-end comparison of Sections~\ref{sec:identities}--\ref{sec:verification};
and the Wiener--Hammerstein real-data run---together with the symbolic
derivations in rational coordinates, the Lean module with its build
configuration, and the per-cell logs behind
Tables~\ref{tab:verification} and~\ref{tab:diagnostic} (the QMC budget ladder and
the lognormal tail-truncation diagnostic included). All random draws are seeded;
the three-path protocol is deterministic except for the explicitly seeded
sampling paths.

\section*{Data and code availability}
All code, the Lean~4 module, and the stored numerical outputs are released as a
public MIT-licensed supplement at
\url{https://github.com/SZabolotnii/Ku-Projection-Framework-code-supplement}
(Appendix~\ref{app:repro}). The Wiener--Hammerstein records analyzed in
Section~\ref{sec:diagnostic} are the public benchmark of
\citet{schoukens2009wienerhammerstein}.

\section*{Declaration of generative AI and AI-assisted technologies}
During the preparation of this work the author used large language model
assistants to improve the language and readability of the manuscript and to
assist with LaTeX formatting, and used Claude Code (Anthropic) to assist in
writing the Lean~4 formalization code released in the reproducibility supplement.
The Lean development is machine-checked by the Lean~4 kernel and Mathlib
(\texttt{sorry}-free, with an audited axiom set; see Appendix~\ref{app:lean}), so
the validity of the formalized statements does not rest on the assisting tool.
After using these tools, the author reviewed and edited all content as needed and
takes full responsibility for the content of the publication.

\FloatBarrier
\bibliographystyle{plainnat}
\bibliography{references}

\end{document}